\documentclass{amsart}
\usepackage[dvips]{graphicx}
\usepackage{amscd}
\usepackage{pstricks}

\theoremstyle{plain}
\newtheorem{theorem}{Theorem}

\newtheorem{lemma}{Lemma}
\newtheorem{definition}{Definition}
\newtheorem{corollary}{Corollary}

\theoremstyle{remark}
\newtheorem{remark}{Remark}

\numberwithin{equation}{section}

\begin{document}
\title[]{A pulse fishery model with closures as function of the catch: Conditions for sustainability}
\author[]{Fernando C\'ordova--Lepe}
\author[]{Rodrigo Del Valle}
\author[]{Gonzalo Robledo}
\address{Facultad de Ciencias B\'asicas, Universidad Cat\'olica del Maule, Talca, Chile. Avenida San Miguel 3605, Casilla 617 Talca, Chile}
\address{Departamento de Matem\'aticas, Facultad de Ciencias, Universidad de
  Chile, Casilla 653, Santiago, Chile}
\email{fcordova@ucm.cl,rvalle@ucm.cl,grobledo@uchile.cl}
\date{October 2011}

\begin{abstract}
We present a model of single species fishery which alternates closed seasons with pulse captures. The novelty is that the length of a 
closed season is determined by the stock size of the last capture. The process is described by a new type of impulsive 
differential equations recently introduced. The main result is a fishing effort threshold which determines either the sustainability of the fishery or
the extinction of the resource.
\end{abstract}
\maketitle

\section{Introduction}
\subsection{Preliminaries}
This work proposes a model of management of \emph{closed seasons} (also named \emph{seasonal closures}) for 
fisheries. The FAO's fisheries glossary\footnote{We refer the reader to www.fao.org/fi/glossary.} defines a closed season of a single marine species 
as \emph{the banning of fishing activity} (\emph{in an area or 
of an entire fishery}) \emph{for a few weeks or months, usually to protect juveniles or spawners}. The closed seasons are implemented by a regulator
authority in order to limit the productive factors in certain area during a specific time interval. 

A fishery process could envisage alternated periods of closures and \emph{open seasons}, where the fishing 
is allowed. The closures are introduced according to bioeconomical necessities and its scope considers several 
spatio-temporal possibilities, leading to concepts as: \emph{seasonal closures--no area closure}, 
\emph{short term area closure--no seasonal closure}, \emph{short term time} and \emph{fully--protected area}, see \emph{e.g.}, \cite{Domeier} for 
details. 

This article will consider the special --but important-- case of bioeconomic processes having the following 
property \cite{Yang-L}: \emph{low frequency / short duration of open seasons} combined with \emph{large magnitude of capture}. In 
this framework, we propose a feedback regulatory po\-li\-cy which defines the lenght of the next closure as a function of the present captured biomass. A 
consequence of this regulation is that under certain fishing effort (this concept will be explained later) threshold, the convergence 
towards a constant length closures ensuring 
the ecological sustaintability of the resource is obtained.

The literature shows many cases of a long--term time 
and area closures. An example is the pacific herring (\emph{Clupea pacificus}) fishery, where
the regulatory measures have included very short open seasons as two hours joined with other inputs and outputs restrictions \cite{Pitcher}. Another
example is given by a village--base management program in Vanuatu, where the stocks should be harvested in a sequence of brief openings interspresed 
with several years closures (see \cite{Johannes1}, \cite{Johannes2} for details). Cases of annual fisheries combining short periods with large 
captures, followed by (comparatively) low ones at the rest, are registered. For instance, the fishery of common sardine (\emph{Strangomera bentincki}) and 
anchovy (\emph{Engraulis ringens}) in the southern Chilean coast has a behavior described by the Figure 1.

\begin{center}
\label{tmouna1-0}
\begin{figure}
\label{impulsos-capturas}
\includegraphics[width=10cm, height=6cm]{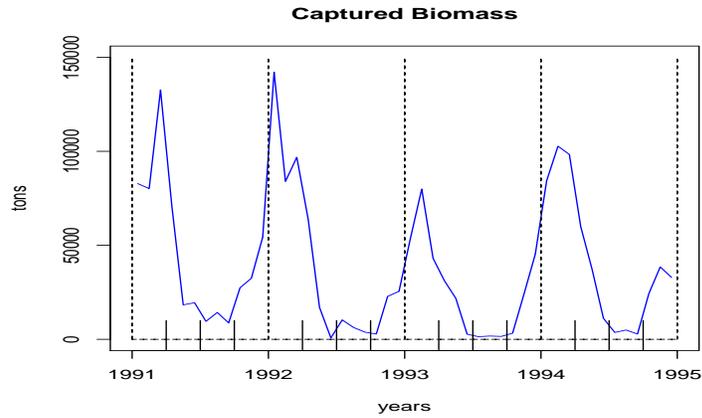}
\caption{Captures of anchovy and common sardine (1991--1995) in southern Chilean coast \cite{Cubillos}. The period with 
the largest fishing mortality is the summer, while the lowest one corresponds to the closed season.} 
\end{figure}
\end{center}

\begin{center}
\label{tmouna1}
\begin{figure}
\label{impulsos-capturas}
\includegraphics[width=10cm, height=6cm]{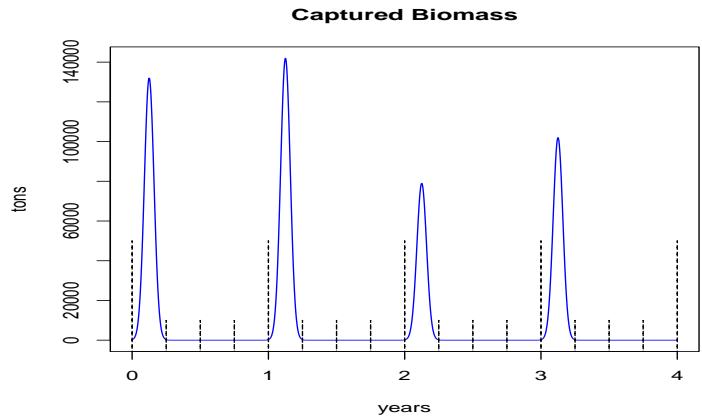}
\caption{Idealized scheme of the previous figure: intense fishing season alternated by absolute 
closures of nine months. If the lenght of openings are small compared with closures and the harvested stock is large, it can 
be considered as a pulse.} 
\end{figure}
\end{center}

The problem of determining the lenght of a closure (in a context of short--term open seasons with intense fishing effort) ensuring the bioeconomical 
sustantability is a complex issue: indeed, short term closures followed by intense captures could induce an overexploitation of the resource. On the other 
hand, long--term closures could have some unexpected drawbacks as: i) \emph{Bio--sustainable economic rent} with negative average \cite{Clark}. ii) Promotion 
of negative \emph{indirect effects} \cite{Bender}, \cite{Higashi} when the resource 
is a top predator. An 
example is given by the closure (1989--1992) of the mollusk \emph{Concholepas concholepas} in the chilean coasts \cite{Stoltz}. iii) The phenomenon known 
as \emph{race for fish}: fishing units try to outdo one another in fishing power and efficiency during the brief openings \cite{Copes}. 


\subsection{Mathematical modeling}
In this paper, we will assume that the fishery process has two different 
time--scales. In the first scale (closed seasons), the growth of a single unstructured marine resource is described by using ordinary 
differential equations. Nevertheless, as the capture has a duration considerably 
shorter than the closures lenght, every open season will be considered as an instant of capture, \emph{i.e.}, a term of a sequence $\{t_{k}\}_{k}$,
this is the second scale (See Figure 2). In consequence, the global process will be described by 
an impulsive differential equation, IDE (see \emph{e.g.}, \cite{Haddad}, \cite{Samoilenko} for details).

IDE equations have been used in the mathematical modeling of processes involving impulsive 
harvesting: \emph{e.g.}, \cite{Bai}, \cite{Berezanski}, \cite{Braverman}, \cite{Dong}, \cite{Zhang-0}, \cite{Zhang-1}, \cite{Zhao-03}, \cite{Zhao} (see 
\cite{Mai-Le} and \cite{Yang-L} for more 
applications of IDE to bioeconomics and ecological processes respectively). In all these references, the following sequence of harvest instants is employed:
\begin{equation}
\label{sequence-1}
t_{k}=k\tau \quad \textnormal{with} \quad k=0,1,\ldots,\quad \textnormal{and} \quad \tau>0,
\end{equation}  
which implies that $\tau=t_{k+1}-t_{k}$ is the time between two consecutive captures.

This paper follows a different approach. Indeed, we introduce a new sequence $\{t_{k}\}_{k}$ of harvest instants, where 
the time for the $(k+1)$--th capture is determined by the amount of the biomass harvested
at $k$--th time (which will be denoted by $x(t_{k})$). This leads to a length of closed season determined by
\begin{equation}
\label{sequence-2}
t_{k+1}-t_{k}=\tau(x(t_{k})),
\end{equation}
where $\tau(\cdot)$ is some function depending on the amount of the biomass captured at $t=t_{k}$. In consequence, the biomass harvested at 
time $t=t_{k}$ will determine the instant of the next capture. The idea is to define a length closure such that a bigger capture leads to a longest closure
and by assuming broad and realistic properties on per capita growth rate and product function, we will find sufficient conditions ensuring a sustainable production, {\it i.e.}, the existence of a globally stable periodic trajectory. 

This approach has been introduced in \cite{Cordova2}, where it is pointed out that
the resulting model is a new type of IDE equations, namely IDE with impulses dependent of time (IDE--IDT) and an introductory theory is presented.

\subsection{Outline}
In section 2, we construct a model of fishery with closed seasons and pulse captures, which is described by an IDE--IDT. In section 3, we study 
some basic properties of the resulting IDE--IDT. The main results concerning the sustainability of the resource
are  presented in section 4. A numerical example is presented in section 5.

\section{The model}
A classical mathematical model of a fishery with closed seasons has to describe the following bioeconomic 
processes: i) The growth of the resource (\emph{e.g.}, a single marine species). ii) The production function. 
iii) The length of the closures.

\subsection{Natural rate of growth}

The growth of the biomass in the close season will be described by the 
ordinary differential equation:
\begin{equation}
\label{natural}
x' (t)  = x(t)\,r(x(t)),  \quad \textnormal{for any} \quad t\in (t_{k},t_{k+1}),
\end{equation}
where $\{t_{k}\}$ denotes an increasing 
sequence of harvest instants.

\noindent{\bf Growth hypotheses (G)} 
\begin{itemize}
\item[\textbf{(G1)}] {\bf Density--dependence.}  The \emph{per capita rate of growth}  $r\colon [0,K] \to [0,+\infty)$ 
is a derivable and strictly decreasing function of the biomass. In addition, we assume $r(0)=r_{0}>0$ and  $r(K)=0$. 
\item[\textbf{(G2)}] {\bf Bounded variation.} The rate  $r\colon [0,K] \to [0,+\infty)$ has lowerly bounded 
derivative, {\it i.e.}, there exists $\rho>0$, such that 
$-\rho \leq r'(x)<0$.
\end{itemize}

\begin{remark}
\label{commentaire-1}
The property of derivability stated in \textbf{(G1)} is a technical assumption. Nevertheless, in \cite{Tanner}, some statistical
results support the negative correlation between per capita rate of growth and biomass. Notice that, $0$ and $K$ are the unique 
equilibria of (\ref{natural}) and $[0,K]$ is a positively invariant set.
\end{remark}

\begin{remark}
\textbf{(G2)} implies that the function $x\mapsto x\,r(x)$ satisfies the local 
Lipschitz condition when $x \in [0,K]$. Hence, the solutions of (\ref{natural}) are 
unique to the right and depend continuously on initial condition to the right.
\end{remark}

An example of growth rates satisfying \textnormal{\textbf{(G)}} are given by the family:
\begin{equation}
\label{gen-ver}
r(x)=r_{0}\Big(1-\Big[\frac{x}{K}\Big]^{\theta}\Big)^{\beta}, \quad \textnormal{with} \quad \theta\geq 1 \quad \textnormal{and}
\quad \beta\geq 1,
\end{equation}
which generalizes the \emph{logistic Verhulst--Pearl per capita rate} (see \emph{e.g.}, \cite{Tsoularis}).

Another example is given by the function \cite{Smith63}:
\begin{equation}
\label{FSmith}
r(x)=r_{0}\frac{K-x}{K+cx}, \quad \textnormal{with} \quad  c>0, 
\end{equation}
which was used, for example, to describe the growth of \emph{Daphnia Magna}.

\subsection{Fishing mortality}
The \emph{fishing mortality} (see \cite[p.102]{Kolding}) is the fraction $F\in [0,1]$ of average population taken by fishing.
In order to estimate it, we emphasize that there are two possible scenarios: an \emph{open season} one, where the fishing is allowed all the time,
and a restricted process (\emph{closed season}), where the fishing is forbidden. 
 
The global fishing process is studied with two time scales: the first one describes the close season 
and only considers the growth of biomass summarized by (\ref{natural}) with assumptions \textbf{(G)}. The second scale concentrates the fishing 
mortality by considering the captures as pulses defined in a sequence of \emph{harvest instants} $\{t_{k}\}$, obtaining:
\begin{equation}
\label{cosecha}
F=\frac{x(t_{k})-x(t_{k}^{+})}{x(t_{k})}=\frac{H(x(t_{k}))}{x(t_{k})},
\end{equation}
where $x(t_{k}^{+})$ is the after $k$--th capture biomass and $H(\cdot)$ describes the \emph{capture function} in the biomass level 
at $t_{k}$. From (\ref{cosecha}), we can see that fishing mortality can be seen as a measure of \emph{catch per unit of biomass} (CPUB).

For convenience, let us define
the \emph{impulse action} $I\colon [0,K]\mapsto [0,K]$, as follows: 
\begin{equation}
\label{action}
I(x)=x-H(x)=(1-F)x.
\end{equation}

In order to relate fishing mortality with the input factors (capital and labor) deployed along the fishing 
process, the continuous modeling literature has introduced the concept of \emph{fishing effort}, which is a rate describing the number of 
boats, traps, hooks, technicians, fishermens, etc., per time (see \emph{e.g.}, \cite{Anderson}, 
\cite{Clark}, \cite{Kolding}). 

In an impulsive modeling framework, if the \emph{punctual fishing effort} is denoted by $E>0$, a bounded scalar measure, it is expected
to describe (\ref{cosecha}) by a functional relation $F=\phi(E,x(t_{k}))$, which, combined with (\ref{natural}) and (\ref{action}), allows to introduce the complementary 
evolution law: 
\begin{equation}
\label{cosecha-1}
x(t^{+})=I(x(t))=\big(1-\phi(E,x(t))\big)x(t), \quad \textnormal{with} \quad t=t_{k}.
\end{equation}

We point out that the practical estimation of $E$ and $\phi(\cdot,\cdot)$ are complicated matters in bioeconomic 
theory and we refer the reader to \cite{Anderson}, \cite{Clark}, \cite{Walters} for details.

\noindent{\bf Harvest hypotheses (H):}
\begin{itemize}
\item[\textbf{(H1)}] {\bf Impulsive action}. $I(\cdot)$ is a derivable and 
increasing function. 
\item[\textbf{(H2)}] {\bf Elasticity}. If $\Delta x \to 0$, then:  
\begin{equation}
\label{elasticity}
\frac{\Delta H/H}{\Delta x /x} \to \frac{H'(x)\, x}{H(x)}\geq 1,
\end{equation}
for any $x \in (0,K]$. This is, a percentage change in the resource biomass determines a percentage change 
bigger than or equal in the captured amount. Notice that, if the {\it yield elasticity respect to the biomass} is bigger or equal than one,
then $I(\cdot)$ is {\it inelastic} or {\it unitary}.
\end{itemize}

\begin{remark}
\label{commentaire-h1}
The property \textbf{(H1)} combined with (\ref{cosecha-1}) says that:
$$
x\frac{\partial \phi}{\partial x}(E,x) \leq 1-\phi(E,x) \quad \textnormal{for any} \quad x\in [0,K].
$$
\end{remark}

\begin{remark}
\label{equivalence-0}
The property \textbf{(H2)} says that a fixed punctual fishing effort is more productive at higher resource availability.
In addition, by using (\ref{action}), we can prove that \textbf{(H2)} is equivalent to
$$
\frac{I'(x)\, x}{I(x)}\leq 1, \quad \textnormal{and} \quad \frac{\partial \phi}{\partial x}(E,x) \geq 0, \quad \textnormal{for any} \quad x \in (0,K].
$$
\end{remark}

\begin{remark}
\label{C-D}
Notice that, the called Cobb--Douglas production function can be interpreted by a fishing mortality $\phi(E,x)=qE^{\alpha}x^{\beta-1}$ 
with $q>0$, $\alpha>0$ and $\beta>0$.  It is easy to see that hypotheses \textbf{(H)} are verified when $\beta\geq 1$. An important case
is the Schaefer assumption \cite{Schaefer}, corresponding to $\beta=1$ and $\alpha=1$, \emph{i.e.}, the parametrization is linear with respect to the fishing effort
and biomass.
\end{remark}

\subsection{Length of the closures}

There exist a third evolution law governing the dynamics, which determine the 
sequence $\{t_{k}\}$ of harvest instants. It is the first order recurrence that follows:
\begin{equation}
\label{recurrencia}
\Delta \, t_{k}=t_{k+1}-t_{k}=\tau(I(x(t_{k}))),
\end{equation}
where $\tau:[0,K^{*}] \to [0,+\infty)$, with $K^{*}=I(K)$.

As it can be observed in (\ref{recurrencia}), the length of 
the next closure, namely $\Delta t_{k}$, is a function $\tau(\cdot)$ of the stock after the $k$--th harvest and allows 
to establish an automatic regulation of the dynamics by closed seasons. Here, we only introduce a theoretical proposal 
and we do not deal with the problems of implementation, for instance, those relating to the 
estimation of data requirred to define the length of the closed seasons.



                


\subsection{The equation model}

The dynamics determined by the combination of equations (\ref{natural}), (\ref{cosecha}) and (\ref{recurrencia}) is formalized by the 
impulsive differential equation:
\begin{equation}
\label{impulsiva}
	\left\{ \begin{array}{llll}
	x' (t)      & = & x(t)r(x(t)),              & t\neq t_k, \\
	x\,(t^{+})   & = & I(x(t)),          & t=t_k,\\ 
	\Delta t_{k}  & = & \tau(I(x(t_k))), & k\geq 0,\\
	\end{array}
	\right.
\end{equation}
where $(t,x)\in [0,+\infty)\times[0,+\infty)$. 

This type of impulsive differential equation is denoted as Impulsive Differential Equations with Impusive Dependent Times (IDE--IDT), which
have been introduced by C\'ordova--Lepe in \cite{Cordova2} and its novelty with respect to classic impulsive differential equations
is that the sequence of impulse instants is determined by the process dynamics: the harvested stock $I(x(t_{k}))$  will determine 
the next harvest time $t_{k+1}$. Indeed, the sequence of harvest times is described by:
\begin{equation}
\label{modelo-ii}
t_{k+1}=t_{k}+\tau(I(x(t_{k}))),\\
\end{equation} 
where the biomass $x(t)$ is abruptly reduced to $x(t^{+})=x(t)-I(x(t))$ at $t=t_{k}$.

There exists several models of pulse harvesting of a renewable resource (not uniquely restricted to 
fisheries) described by impulsive equations, \emph{e.g.},: \cite{Bai}, \cite{Berezanski}, \cite{Braverman}, \cite{Cor-Pin} and \cite{Zhang-0} consider 
a resource with logistic growth rate, \cite{Zhao} considers a generalized logistic growth.
In addition, Gompertz models (which, not satisfy \textbf{(G2)}) have been studied in \cite{Braverman}, \cite{Dong}.  Nevertheless, 
these works consider a fixed time between two harvest processes, which is equivalent to
consider $\tau(\cdot)$ as a constant function.

We point out that impulsive models having sequences similar to (\ref{modelo-ii}) have also been 
introduced in \cite{Karafyllis} by Karafyllis in an hybrid control theory
setting and are named \emph{hybrid systems with sampling partition generated by the system}.

\section{The impulsive system (\ref{impulsiva})}
Given a first harvest time $t_{0}\in \mathbb{R}$ and a biomass level $x_{0}\in [0,K]$, then the existence, uniqueness 
and continuability of the solution of (\ref{impulsiva}), with initial condition $(t_{0},x_{0})$, can be 
deduced from \cite{Cordova2}. Indeed, we know any solution is a piecewise
continuous function having first kind discontinuities at the harvest instants $t=t_{k}$ ($k=0,1,2,\cdots$). In addition, we point out that 
different initial conditions will determine different sequences of harvest instants.

\subsection{Basic properties}
In the study of (\ref{impulsiva}), it will be necessary to consider the initial value ordinary associated problem:
\begin{equation}
\label{ode-0}
	\left\{ \begin{array}{l}
	z'(t)    \hspace{0.13cm}   =  z(t)r(z(t)),               \\
	z(\sigma_{0}) =  v, \quad \textnormal{with} \quad (\sigma_{0},v) \in \mathbb{R}\times [0,K].                 
	\end{array}
	\right.
\end{equation}

\begin{definition}
\label{faux1}
The unique solution of \textnormal{(\ref{ode-0})} will be denoted by $t\mapsto\varphi(t;\sigma_{0},v)$, for 
any $t\geq \sigma_{0}$, and $\varphi(\sigma_{0};\sigma_{0},v)=v$.
\end{definition}

Observe that given $v \in [0,K]$, the function $\varphi(\cdot;\sigma_{0},v) \colon [\sigma_{0},+\infty)\to [0,K]$ satisfies:
\begin{equation}
\label{prop-can}
\varphi(t;\sigma_{0},0)=0 \quad 
\textnormal{and}
\quad
\frac{\partial \varphi}{\partial v}(t;\sigma_{0},v)\geq 0.
\end{equation}

Let $x(\cdot)$ be the solution of (\ref{impulsiva}) with initial condition $(t_{0},x_{0})$, which determines
the sequence $\{(t_{k},x(t_{k}))\}_{k}$. Since (\ref{impulsiva}) is an ODE on $(t_{k},t_{k+1}]$, we can deduce that
$x(\cdot)$ coincides with the unique solution $\varphi(\cdot,t_{k},I(x(t_{k})))$ of (\ref{ode-0}) on $(t_{k},t_{k+1}]$. 

By using Definition \ref{faux1}, it follows that
$$
\varphi(\sigma+t_{k}; t_{k},I(x(t_{k})))=
I(x(t_{k}))\exp\Big(\int_{t_{k}}^{\sigma+t_{k}}r\big[\varphi(s;t_{k},I(x(t_{k})))\big]\,ds\Big),
$$
$$
\hspace{3.65cm}=I(x(t_{k}))\exp\Big(\int_{0}^{\sigma}r\big[\varphi(s+t_{k};t_{k},I(x(t_{k})))\big]\,ds\Big).
$$

Finally, uniqueness of solutions implies $\varphi(s+t_{k};t_{k},I(x(t_{k})))=\varphi(s;0,I(x(t_{k})))$, which leads to: 
\begin{equation}
\label{pitaron-ode}
\varphi(\sigma+t_{k}; t_{k},I(x(t_{k})))=
I(x(t_{k}))\exp\Big(\int_{0}^{\sigma}r\big[\varphi(s;0,I(x(t_{k})))\big]\,ds\Big),
\end{equation}
for any $\sigma \in [0,\tau(I(x(t_{k}))]$.

By using (\ref{modelo-ii}), it follows that at $\sigma=\tau(I(x(t_{k})))$ (\emph{i.e.}, at $t=t_{k+1}$), the solutions of
(\ref{impulsiva}) satisfy the one dimensional map:
\begin{equation}
\label{pitaron-ode1}
x(t_{k+1})=f(x(t_{k})),
\end{equation}
where the function $f\colon [0,K]\to [0,K]$ is described as follows:
\begin{equation}
\label{recurrencia-0}
f(x)=F(I(x)), \quad \textnormal{with} \quad F(y)=y \, \exp \Big( \int_0^{\tau(y)}r[\varphi(s;0,y)]ds  \Big).
\end{equation}

Notice that Eq.(\ref{cosecha-1}) implies $f(0)=0$. This fact will have important consequences when studying the asymptotic behavior of
(\ref{impulsiva}).

\subsection{Special solutions of (\ref{impulsiva}) and bio-economic interpretation}
Let us introduce the straightforward result:
\begin{lemma}
\label{pitaron-system}
Any positive fixed point $u^{*}\in (0,K]$ of the map \textnormal{(\ref{pitaron-ode1})} defines a piecewise--continuous 
$\tau(I(u^{*}))$--periodic 
solution $t\mapsto u^{*}(t)$ of \textnormal{(\ref{impulsiva})}, namely, the $u^{*}$--associated solution. In addition, the fixed point $u^{*}=0$ of the 
map \textnormal{(\ref{pitaron-ode1})} defines a constant null solution of  \textnormal{(\ref{impulsiva})}.
\end{lemma}

The existence of a $\tau(I(u^{*}))$--periodic solution can be interpreted as a fishery strategy with harvest instants uniformly 
distributed in time.  There exists different stability definitions for these solutions (see \emph{e.g.}, \cite{Cordova2} and \cite{Karafyllis}). In this context,
we will follow the ideas stated in \cite{Cordova-Dv-R}:

\begin{definition}
\label{fplas}
The $u^{*}$--associated solution of \textnormal{(\ref{pitaron-ode1})} is locally asymptotically stable if there exists $\delta>0$ such that 
for any solution $t\mapsto x(t)$ of \textnormal{(\ref{impulsiva})} 
with initial condition $x(0)$ satisfying $|x(0)-u^{*}|<\delta$, it follows that:
$$
\lim\limits_{k\to +\infty}|x(t_{k})-u^{*}|=0,
$$
where $\{t_{k}\}$ is the corresponding sequence of harvest instants associated to $x(0)$.
\end{definition}

\begin{definition}
\label{fpgas}
The $u^{*}$--associated solution of \textnormal{(\ref{pitaron-ode1})} is globally asymptotically stable if for any solution 
$t\mapsto x(t)$ of \textnormal{(\ref{impulsiva})} 
with initial condition $x(0)>0$, it follows that:
$$
\lim\limits_{k\to +\infty}|x(t_{k})-u^{*}|=0.
$$
\end{definition}

Observe that the asymptotic stability of a $\tau(I(u^{*}))$--periodic solution implies the ecological sustainability of the fishery. On the other
hand, the asymptotic stability of the null solution implies the future resource extinction.

\section{Sustainability conditions}
\subsection{General result}

\begin{theorem}
\label{mainresult}
Let us assume that \textnormal{\textbf{(G)},\textbf{(H)}} and the closed season hypotheses:
\begin{itemize}
\item[\textbf{(C1)}] {\bf Growth type.}
The function $\tau \colon [0,K^{*}] \to [0,+\infty)$ is derivable and decreasing, such that $\tau(K^{*})=0$. 

\item[\textbf{(C2)}] {\bf Initial condition}. The initial value $\tau_{0}=\tau(0)$ satisfies 
                \begin{equation}
                \label{taumax}
                \tau_{0} \leq \frac{1}{\rho K +r_{0}}\ln\Big(1+\frac{1}{1-\phi(E,0)}\Big(1+\frac{r_{0}}{\rho K}\Big)\Big).
                \end{equation}
                
\item[\textbf{(C3)}] {\bf Main condition}. The following inequality: 
	         \begin{equation}
                \label{taucota}
		  |\tau(z_{2})-\tau(z_{1})| < \frac{1}{r_{0}}\left\{\ln \left(\frac{z_{2}}{z_{1}}\right)-m(z_{2}-z_{1})   \right\},
                \end{equation}
	         is verified for $0<z_{1}<z_{2}<I(K)=K^{*}$ and $m=\rho[e^{\alpha\tau_{0}}-1]/\alpha$, with $\alpha=\rho K+r_{0}$.
\end{itemize}
are satisfied.

\vspace{0.01cm}
Then:
\begin{itemize}
\item[i)] If $e^{\tau_{0} r_{0}}(1-\phi(E,0)) < 1$, then the resource--free solution of \textnormal{(\ref{impulsiva})} is globally 
asymptotically stable \textnormal{(}extinction case\textnormal{)}. 
\item[ii)] If $e^{\tau_{0} r_{0}}(1-\phi(E,0)) > 1$, then there exists a unique initial condition $x^{*}=f(x^{*}) \in (0,K)$ 
--with $f(\cdot)$ given by \textnormal{(\ref{pitaron-ode1})}-- defining 
a $\tau(I(x^{*}))$--periodic globally asymptotically stable trajectory \textnormal{(}sustainable case\textnormal{)}.  
\end{itemize}
\end{theorem}


\begin{remark}
\label{TAU}
Notice that \textbf{(C3)} gives us a range of graphic possibilities for 
the choice of function $\tau(\cdot)$. Moreover, (\ref{taumax}) implies that the right side of (\ref{taucota}) is non 
negative for any $z \in [0,K^{*}]$. 
\end{remark}

\begin{remark}
If \textbf{(C2)} is verified, we can see that $e^{\tau_{0}\,r_{0}}(1-\phi(E,0) ) > 1$ implies
$$
\frac{1}{r_{0}}|\ln(1-\phi(E,0))|<\tau_{0}\leq  \frac{1}{\rho K +r_{0}}\ln\Big(1+\frac{1}{1-\phi(E,0)}\Big(1+\frac{r_{0}}{\rho K}\Big)\Big).
$$
In consequence, the left inequality says that there exists a trade--off between the fishing effort $E$ and the maximal lenght of a closure $\tau_{0}$ ensuring the fishery sustaintability. In addition, observe that the inequality stated above has sense only when $\phi(E,0)\in [0,\phi^{*})\subset [0,1]$, which add a complementary constraint for the fishing effort. 
\end{remark}

\begin{proof}
The asymptotic behavior of (\ref{impulsiva}) is determined by the map (\ref{pitaron-ode1}). 

Now, we will verify that $f\colon [0,K] \to \mathbb{R}$ satisfies the following properties.
\begin{itemize}
 \item[a)] The map $f$ is derivable and $f(0)=0$,
 \item[b)] The map $f$ is increasing and $f(K)<K$,
 \item[c)] For any $x\in (0,K]$ it follows that:
$$
0\leq \frac{xf'(x)}{f(x)}<1.
$$
\end{itemize}

Indeed, a) is straightforward consequence from (\ref{cosecha-1}) combined with $f'(x)=F'(I(x))\,I'(x)$, where $F'(y)$
is defined by:
\begin{displaymath}
\exp\Big(\int_{0}^{\tau(y)}r[\varphi(s;0,y)]\,ds\Big)
\Big\{1+y \, \Big(r[\varphi(\tau(y);0,y)] \, \tau'(y) + \int_{0}^{\tau(y)} \frac{\partial r}{\partial y}[\varphi(s;0,y)]\,ds   \Big)\Big\}.
\end{displaymath}

Let us verify that $f'(0)=e^{\tau_{0}r_{0}}I'(0)$ is consequence from \textbf{(G1)},\textbf{(H1)},\textbf{(C2)} and (\ref{prop-can}). Now, by 
using \textbf{(H1)}, we observe that b) follows if $F'(y)>0$ for any $y\in [0,K]$. When dropping the 
exponential factor of $F'(\cdot)$, we only have to prove that
\begin{equation}
\label{toprove}
1+y \, \Big(r[\varphi(\tau(y);0,y)] \, \tau'(y) + \int_{0}^{\tau(y)} \frac{\partial r}{\partial y}[\varphi(s;0,y)]\,ds   \Big) >0
\end{equation}
for any $y \in [0,K]$.

By hypotheses \textbf{(G)} and \textbf{(C1)}, combined with $\frac{\partial \varphi}{\partial y}(s;0,y)>0$ for 
any $s \in [0,\tau(y)]$, we can observe that inequality (\ref{toprove}) can be deduced from:
\begin{equation}
\label{toprove2}
1>y \, \Big(r_{0} \, |\tau'(y)| + \rho \,\int_{0}^{\tau(y)} \frac{\partial \varphi}{\partial y}(s;0,y)\,ds   \Big).
\end{equation}

By integral representation of $\varphi(s;0,y+h)$ and $\varphi(s;0,y)$, with $s \in [0, \tau(y)]$ and $h>0$ sufficientl small, the 
Gronwall inequality implies that
$$
\left|\varphi(s;0,y+h) - \varphi(s;0,y)\right| \leq |h|\,e^{(\rho\,K+r_{0})s},
$$
for any $s \in [0,\tau(y)]$.  Indeed, $\frac{\partial \varphi}{\partial y}(s;0,y) \leq e^{(\rho K +r_{0})s}$, with $s \in [0,\tau(y)]$. Therefore, 
we can reduce our proof to demand the condition that follows:
\begin{equation}
\label{tauequation}
y \left( \frac{\rho}{\alpha}\left[ e^{\alpha \tau(y)}-1   \right]-r_{0}\, \tau'(y) \right)<1, \quad \textnormal{where} \quad \alpha=\rho K+r_{0}.
\end{equation}

We point out that \textbf{(C3)} is equivalent to (\ref{tauequation}). Indeed, when replacing $z_{1}$ and $z_{2}$ by $y$ and $y+h$ ($h>0$) respectively,
(\ref{tauequation}) is obtained by letting $h\to 0$. Inversely, since $\tau(y) < \tau_{0}=\tau(0)$, the inequality (\ref{taucota}) is obtained by
integrating (\ref{tauequation}) on $[0,K^{*}]$, with $K^{*}=I(K)$.

To prove that the right side of (\ref{taucota}) is greater than zero, we 
observe that
$$
\inf\Big\{ \frac{\ln(K^{*})-\ln(z)}{K^{*}-z} \colon \,\,z \in (0,K^{*}]\Big\}=\frac{1}{K^{*}},
$$ 
and $1/K^{*}>m$ is  
equivalent to \textbf{(C2)}. Finally, observe 
that \textbf{(H)} and \textbf{(C1)} imply $f(K)=I(K)<K$ and property b) follows.

The property c) is equivalent to $x^{2}(f(x)/x)'=F'(I(x))I'(x)x-F(I(x))<0$ for any $x\in (0,K]$. This is verified if and only if for 
any $x \in (0,K]$, it follows that:
\begin{equation}
\label{proelasticidad}
x\, I'(x) \,\left[1+I(x) W'(I(x)) \right] \exp(W(I(x))) \leq I(x) \, \exp(W(I(x)))
\end{equation}
where $W(I(x))$ is defined by
$$
W(I(x))=\int_{[0,\tau(I(x))]}r[\varphi(s;0,I(x))]\, ds,  \quad  \textnormal{for} \quad x \in [0,K]. 
$$

Let us recall that 
$$
W'(y)= r[\varphi(\tau(y);0,y)] \, \tau'(y) + \int_{0}^{\tau(y)} \frac{\partial r}{\partial y}[\varphi(s;0,y)]\,ds   <0, \quad
\textnormal{for} \quad y\in (0,K^{*}].
$$

By using this inequality, combined with (\ref{toprove}) and Remark \ref{equivalence-0}, we can deduce that:
\begin{equation}
\label{preelaticidad}
\frac{I'(x)\,x}{I(x)} \leq 1< \frac{1}{1-I(x)|W'(I(x))|}, \quad \textnormal{with} \quad x \in (0,K].
\end{equation}

By applying Lemma \ref{discr3} (see Appendix), there are two posibilities for system (\ref{pitaron-ode1}) according the sign of 
$$
1-f'(0)=1-e^{\tau_{0}r_{0}}I'(0).
$$
By Eq.(\ref{cosecha-1}), we can express the threshold condition for case (a) $f'(0) < 1$ as $e^{\tau_{0}\,r_{0}}(1-\phi(E,0)) < 1$ and 
for case (b) $f'(0)>1$ as $e^{\tau_{0}\,r_{0}}(1-\phi(E,0) ) > 1$. So that, the result follows.
\end{proof}

\subsection{Application to logistic growth}
Notice that, in some cases, the one--di\-men\-sional map (\ref{pitaron-ode1}) associated to the system (\ref{impulsiva})
can be defined explicitly and the previous result improved. Indeed, let us consider a marine species with logistic growth, whose exploitation is described by: 
\begin{equation}
\label{dugma}
	\left\{ \begin{array}{llll}
	x' (t)      & = & rx(t)\Big(1-\frac{\displaystyle x(t)}{\displaystyle K}\Big),              & t\neq t_k, \\
	x\,(t^{+})   & = & I(x(t)),          & t=t_k,\\
	\Delta t_{k}  & = & \tau(I(x(t_k))), & k\geq 0.\\
	\end{array}
	\right.
\end{equation}

\begin{corollary}
\label{verhuls-pearl}
Let us assume that the impulse action $I(\cdot)$ satisfies \textnormal{\textbf{(H)}} and the close season satisfies \textnormal{\textbf{(C1)}}
and 
\begin{itemize}
\item[\textbf{(C3')}] {\bf Closure condition}. The following inequality: 
	         \begin{displaymath}
		  |\tau(z_{2})-\tau(z_{1})| < \frac{1}{r_{0}}\left\{\ln \left(\frac{z_{2}}{z_{1}}\right)+\ln\Big(\frac{K-z_{1}}{K-z_{2}}\Big)   \right\},
                \end{displaymath}
	         is verified for $0<z_{1}<z_{2}<I(K)=K^{*}$.
\end{itemize}
Then:
\begin{itemize}
\item[i)] If $e^{r\tau_{0}}(1-\phi(E,0)) < 1$, then the resource--free solution of \textnormal{(\ref{dugma})} is globally 
asymptotically stable \textnormal{(}extinction case\textnormal{)}. 
\item[ii)] If $e^{r\tau_{0}}(1-\phi(E,0)) > 1$, then there exists a unique initial condition $x^{*}=f(x^{*}) \in (0,K)$ defining 
a $\tau(I(x^{*}))$--periodic globally asymptotically stable trajectory \textnormal{(}sustainable case\textnormal{)}.  
\end{itemize}
\end{corollary}

\begin{proof}
A simple computation shows that (\ref{pitaron-ode1}) is equivalent to the one--dimensional map:
\begin{equation}
\label{logistico}
x(t_{k+1})=f(x(t_{k}))=\frac{KI(x(t_{k}))}{I(x(t_{k}))+[K-I(x(t_{k}))]e^{-r\tau(I(x(t_{k})))}}.
\end{equation}

We will verify that the map (\ref{logistico}) satisfies the assumptions of Lemma \ref{discr3} (see Appendix). Firstly, notice that 
$f(\cdot)$ can be described as follows: 
$$
f(x)=F(I(x)), \quad \textnormal{with} \quad F(u)=\frac{Ku}{u+[K-u]e^{-r\tau(u)}},
$$ 
and by using \textbf{(H)}, it follows that $f(\cdot)$ is derivable and $f(0)=0$. 

Secondly, observe that $f'(x)=F'(I(x))I'(x)$. As in the proof
of Theorem \ref{mainresult}, it follows that $F'(u)>0$ (with $0\leq u \leq I(K)=K^{*}$) if and only if \textbf{(C3')} is verified. By using this
fact, combined with \textbf{(H1)}, it follows that $f(\cdot)$ is increasing and $f(K)=F(I(K))<K$.

Finally, from \textbf{(H2)} combined with Remark \ref{equivalence-0}, it follows that
$$
0 \leq \frac{xf'(x)}{f(x)}=\frac{xI'(x)F'(I(x))}{F(I(x))}=\frac{xI'(x)}{I(x)}\frac{F'(I(x))I(x)}{F(I(x))}\leq \frac{F'(I(x))I(x)}{F(I(x))}.
$$

By using \textbf{(C1)}, it is not difficult to show that $uF'(u)<F(u)$ for any $u \in (0,I(K)]$. This fact, combined with the last above inequality
implies that $0\leq xf'(x)/f(x)<1$ for any $x\in (0,K]$. Now, as $f'(0)=e^{r\tau_{0}}(1-\phi(E,0))$, the result is a direct consequence from Lemma \ref{discr3}.
\end{proof}

\begin{remark}$ $
\label{r-sof}
\begin{itemize}
\item[i)] The assumption \textbf{(C3')} is equivalent to the differential inequality $K>-ru(K-u)\tau'(u)$ for any $0\leq u \leq I(K)$, which
furnishes a way to design admissible functions $\tau(\cdot)$ describing the lengh of open seasons.
\item[ii)] In addition, it is easy to verify that in this case there are no explicit restrictions for $\tau_{0}$ as stated in \textbf{(C2)} (see also Remark \ref{TAU}). 
\end{itemize}
\end{remark}

\section{Example}
Let us consider a fishery with: biomass growth described by the logistic equation, fishing mortality satisfying Schaefer 
assumption, \emph{i.e.}, $\phi(E,x)=qE$ and closures having lengths determined by the linear function $\tau\colon [0,K^{*}]\to [0,+\infty)$:
\begin{equation}
\label{tauexample}
\tau(z)=a(K^{*}-z), \quad \textnormal{with} \quad a>0.
\end{equation}

By using remarks \ref{commentaire-1} and \ref{C-D}, it follows that hypotheses \textbf{(G)} and \textbf{(H)} are
satisfied. In addition, observe that assumption \textbf{(C1)} is satisfied since $\tau(\cdot)$ is 
strictly decreasing and $\tau(K^{*})=0$. Moreover, notice that: 
$$
-ru(K-u)\tau'(u)= aru(K-u) \leq ar K^{2}/4,
$$
and by using statement i) from Remark \ref{r-sof}, it follows that \textbf{(C3')} is satisfied if $a<4/(rK)$.

On the other hand, as $K^{*}=(1-qE)K$ and $\tau_{0}=\tau(0)=a(1-qE)K$, we obtain the threshold:
\begin{equation}
\label{umbral}
\mathcal{E}(qE)=e^{r\tau_{0}}(1-\phi(E,0))=e^{arK(1-qE)}(1-qE),
\end{equation}
and by Corollary \ref{verhuls-pearl}, it follows that the resource--free solution of (\ref{dugma}) is globally asymptotically stable 
if $\mathcal{E}(qE)<1$. Similarly, there exists a globally asymptotically stable periodic solution if $\mathcal{E}(qE)>1$.

\vspace{5mm}
In order to illustrate some properties of the set of punctual fishing efforts $(0,E^{*})$ ($\mathcal{E}(qE^{*})=1$) ensuring sustainability, let us represent the
slope of (\ref{tauexample}) as follows:
$$
a=\frac{4}{rK}\eta, \quad \textnormal{with}  \quad 0<\eta<1.
$$

Notice that, to find $E^{*}$ is equivalent 
to find the unique fixed point $w^{*}=1-qE^{*}$ of the map $w\mapsto e^{-4\eta w}$. In this case, $E^{*}$ is dependent of the parameter 
$\eta \in (0,1)$, which is positively related to $\tau_{0}$.

It is straightforward to verify that the function $\eta \mapsto E^{*}$ is increasing and concave. This implies that lower values of maximal closure 
length $\tau_{0}$ leads to narrow ranges of sustainable punctual effort $(0,E^{*})$.

We illustrate this previous results by using the numerical methods developed by Del-Valle \cite{Rdv}. The following parameters are employed:
\begin{equation}
\label{parametros}
r_{0}=0.05 \hspace{0.05cm} \textnormal{[time$^{-1}$]}, \quad K=1000 \hspace{0.05cm} \textnormal{[tons]}, \quad q=1 \quad \textnormal{and} \quad \eta=0.25.  
\end{equation}

These values determine $w^{*} \sim 0.567143$ with a respective $E^{*}\sim 0.432857$. The figure
shows the biomass curves associated to the following punctual fishing efforts:
 \begin{equation}
\begin{array}{|c|c|c|c|c|}
\hline
\mbox{Effort} & \mbox{Case 1} & \mbox{Case 2} & \mbox{Case 3} & \mbox{Case 4} \\
\hline
\mbox{$E$}        & 0.1 & 0.4 & 0.4329 & 0.5\\
\hline
\end{array}
\end{equation}

\begin{center}
\begin{figure}
\label{Fig2}
\includegraphics[width=10cm, height=6cm]{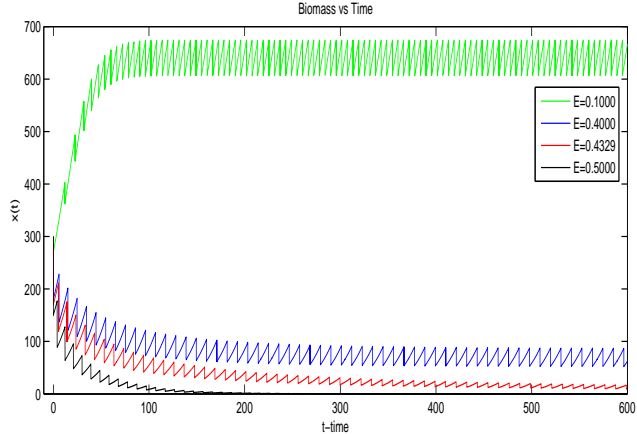}
\caption{Numerical solution of (\ref{dugma}) with Schaefer assumption, closure length defined by (\ref{tauexample}) and parameters (\ref{parametros}). The 
punctual fishing efforts ensuring sustainability satisfy $E\in (0,0.432857)$.} 
\end{figure}
\end{center}

The Figure 3 shows the evolution of the biomass by considering an initial condition $x(0)=300$ tons and four different punctual fishing efforts. As 
stated before, we can see an extinction scenario for any punctual effort bigger than $E^{*}$ (this is the case for $E=0.5$). Finally, it can be observed 
that the lenght of the closure is an increasing function of the punctual effort $E$.





\section{Discussion}
The  mains results (Theorem \ref{mainresult} and Corollary \ref{verhuls-pearl}) propose sufficient conditions ensuring the ecological sustainability
of a simple fishery model with impulsive captures, which can be seen as the trade--off between the fishing effort $E$ and the 
maximal closure length $\tau_{0}$. The novelty (compared with harvest instants uniformly distributed in time) is to allow a 
variable length of closures, which has social and economic consequences in a short term. 

The proof of Theorem \ref{mainresult} is carried out by constructing a one--dimensional map (\ref{pitaron-ode1}), whose 
asymptotic behavior inherits crucial features of the IDE--IDT equation. In this context, this approach could be extended in several 
ways. First, notice that assumptions \textbf{(C2)} and \textbf{(C3)} imply that the map (\ref{pitaron-ode1}) is strictly increasing. Whe think that it is possible to 
consider more general maps and obtain less restrictive conditions ensuring convergence towards a fixed point. This remains a future problem
and its main difficulty is the hybrid nature of IDE--IDT equations.

Provided that the fishery is sustainable, an important open problem is to find first order conditions on the punctual fishing effort $E$ that maximizes:
$$
\frac{\phi(E,x^{*}(E))x^{*}(E)}{\tau([1-\phi(E,x^{*}(E))]x^{*}(E))},
$$
the sustainable production per unit time. In other words, the task is the Maximun Sustainable Yield (MSY) problem.

Another extension of this work would be to consider the logistic case (\ref{dugma}), replacing $r>0$ by a continuous 
function $r\colon \mathbb{R}\to \mathbb{R}_{+}$. This problem is interesting by a bioeconomic point of view since periodic and Bohr 
almost periodic functions provide a good tool in order to modeling birth rates with seasonal behavior.

On the other hand, to consider the logistic model (\ref{dugma}) with almost periodic perturbations has mathematical interest 
in itself. Indeed, if $r(\cdot)$ is a positive Bohr almost periodic function and $\tau(\cdot)$ is a constant function, it can be proved that there exist  
a unique almost periodic solution (in the sense of Samolienko and Perestyuk \cite{Samoilenko}), which is globally asymptotically stable. Nevertheless,
there are neither equivalent results nor a study of asymptotic properties when $\tau(\cdot)$ is not constant. A careful development of the qualitative theory
for IDE--IDT seems essential to cope with this kind of problems.

\section*{Appendix}
The following result plays a key role in the proof of Theorem \ref{mainresult}:
\begin{lemma}
\label{discr3}
Let us consider the one dimensional map:
\begin{equation}
\label{dugma-y}
x_{n+1}=f(x_{n}), \quad x_{0}\in [0,K],
\end{equation}
where the function $f\colon [0,K]\mapsto [0,f(K)]\subset [0,K)$ satisfies the following properties:
\begin{itemize}
 \item[a)] $f$ is derivable and $f(0)=0$,
 \item[b)]  $f$ is increasing and $f(K)<K$,
 \item[c)] For any $x\in (0,K]$ it follows that:
$$
0\leq \frac{xf'(x)}{f(x)}<1.
$$
\end{itemize}

\begin{itemize}
\item[i)] If $f'(0)<1$ then it follows that $\lim\limits_{n\to +\infty}x_{n}=0$.
\item[ii)] If $f'(0)>1$ then there exists a unique fixed point $x^{*}\in (0,K)$ a it follows that $\lim\limits_{n\to +\infty}x_{n}=x^{*}$ for any
$x_{0}\in (0,K]$.
\end{itemize}
\end{lemma}

\begin{proof}
\emph{Case $f'(0)<1$}: we can verify that there exists $\delta>0$ such that $f(x)<x$ for any $x\in (0,\delta)$. Let us denote by $A$, the set
of positive fixed points of $f$. If $A=\emptyset$, by using continuity of $f(\cdot)$ it follows that $f(x)<x$ for any $x\in (0,K]$.

If $A\neq \emptyset$, let us define $\xi =\inf\{A\}$. It is straightforward to verify that:
\begin{equation}
\label{mishvaa}
\xi>0 \quad \textnormal{and} \quad f(x)<x \quad \textnormal{for any} \quad x\in (0,\xi).
\end{equation}

Now, by using c) we can deduce that $0\leq f'(\xi)<1$, which implies the existence of 
$\delta>0$ such that $f(x)\geq x$ for any $x\in (\xi-\delta,\xi)$, obtaining a contradiction
with (\ref{mishvaa}).

In consequence $0$ is the unique fixed point in $[0,K]$ and $f(x)<x$ implies that $\{x_{n}\}$ is strictly descreasing
and lowerly bounded. Finally, the convergence towards $x=0$ follows from uniqueness of the fixed point.

\emph{Case $f'(0)>1$}: We can verify that there exists $\delta>0$ such that $f(x)>x$ for 
any $x\in (0,\delta)$. On the other hand, $f(K)<K$ 
combined with continuity of $f$ imply the existence of a fixed point $x^{*}\in (0,K)$ minimal with this property. Hence,
it follows that $f(x)>x$ for any $x\in (0,x^{*})$ and (c) implies that $0\leq f'(x^{*})<1$.  

In order to prove the uniqueness of $x^{*}$, let us define $g(x)=f(x+x^{*})-x^{*}$ and observe that
$g\colon [0,K-x^{*}]\mapsto [0,f(K)-x^{*}]\subset [0,K-x^{*})$ and $g'(0)<1$. Hence, the non--existence
of fixed points on $(x^{*},K)$ and the property $f(x)<x$ for any $x\in (x^{*},K]$ follows as in the previous case.

If $x_{0}\in (x^{*},K]$ it follows that $\{x_{n}\}$ is strictly descreasing and lowerly bounded by $x^{*}$. On the other hand,  if
$x_{0}\in (0,x^{*})$ it follows that $\{x_{n}\}$ is strictly increasing and upperly bounded by $x^{*}$. The convergence towards $x^{*}$
follows from the uniqueness of the fixed point in $(0,K)$.
\end{proof}

\noindent \emph{Example}: The function $f\colon [0,K]\to \mathbb{R}$ defined by $f(x)=\lambda x/(a+x)$ with $a>1$ and $\lambda<2a$ satisfies
straightforwardly the properties a) and b) by choosing $K>a$. Finally, observe that:
$$
0\leq \frac{xf'(x)}{f(x)}=\frac{1}{a+x}<1, \quad \textnormal{for any} \quad x\geq 0,
$$ 
and c) is verified. 

Hence, if $\lambda<a$ (\emph{i.e.} $f'(0)<1$), it follows that the sequence $\{x_{n}\}_{n}$ defined recursively by (\ref{dugma-y}) with $x_{0}\in (0,K]$ 
is convergent to $0$. Otherwise, if $\lambda \in (a,2a)$ (\emph{i.e}, $f'(0)>1$), then the sequence is convergent to $x^{*}=a-\lambda$. 

\vspace{3mm}

\noindent \textbf{Ackowledgements} We would like to express our thanks to professor Luis 
Cubillos (Universidad de Concepci\'on -- Chile) for the data contained in Figure 1. The first author acknowledges the support of Direcci\'on de Investigaci\'on UMCE.

\end{document}